\documentclass[10pt]{article}

\usepackage[latin1]{inputenc}

\usepackage{amstext,amsmath,amssymb,amsfonts,bbm,amsmath}
\usepackage{authblk}
\usepackage{caption} 
\usepackage{epsfig}  
\usepackage{xcolor} 
\usepackage{graphicx} 
\usepackage{braket}  
\usepackage{amsthm}   
\usepackage{tikz}
\usepackage{enumitem}

\usepackage{hyperref}
\hypersetup{
    colorlinks=false,
    menubordercolor=red,
    linkbordercolor=blue
}

\usepackage{cite}

\allowdisplaybreaks[4]

\usepackage{float}  
\usepackage{subfig}

\usepackage[lmargin=90pt,rmargin=90pt,tmargin=100pt,bmargin=90pt,headheight=26pt,
]{geometry}

\definecolor{blueblue}{rgb}{0.4, 0.6, 0.8} 

\DeclareMathOperator{\Tr}{Tr}

\newcommand{\be}{\begin{equation}}
\newcommand{\ee}{\end{equation}}

\newcommand{\bM}{\mathbf{M}}

\newcommand{\Pro}{\mathbb{P}}
\newcommand{\Exp}{\mathbb{E}}
\newcommand{\cE}{\mathcal{E}}

\numberwithin{equation}{section}

\newtheorem{proposition}{Proposition}
\newtheorem{definition}{Definition}
\newtheorem{lemma}{Lemma}
\newtheorem{theorem}{Theorem}

\theoremstyle{remark}

\begin{document}

\title{\bf The large $N$ factorization does not hold
for arbitrary multi-trace observables in random tensors}

 \author[1]{Razvan Gurau}
 \author[2]{Felix Joos}
 \author[3]{Benjamin Sudakov}
 \affil[1]{\normalsize\it Heidelberg University, Institut f\"ur Theoretische Physik, Philosophenweg 19, 69120 Heidelberg, Germany.
 \authorcr email: gurau@thphys.uni-heidelberg.de \authorcr \hfill}

 \affil[2]{\normalsize\it Heidelberg University, Institut f\"ur Informatik, INF 205, 69120 Heidelberg, Germany.
 \authorcr email: joos@uni-heidelberg.de \authorcr \hfill}
\affil[3]{\normalsize\it Department of Mathematics, ETH, Z\"urich, Switzerland. 
\authorcr email: benjamin.sudakov@math.ethz.ch \authorcr \hfill}

\date{}

\maketitle

\hrule\bigskip

\begin{abstract}
We consider real tensors of order $D$,
that is $D$-dimensional arrays of real numbers $T_{a^1a^2 \dots a^D}$, where each index $a^c$ can take $N$ values. The tensor entries $T_{a^1a^2 \dots a^D}$ have no symmetry properties under permutations of the indices. The invariant polynomials built out of the tensor entries are called trace invariants.

We prove that for a Gaussian random tensor with $D\ge 3$ indices (that is such that the entries $T_{a^1a^2 \dots a^D}$ are independent identically distributed Gaussian random variables) the cumulant, or connected expectation, of a product of trace invariants is \emph{not always} suppressed in scaling in $N$ with respect to the product of the expectations of the individual invariants. Said otherwise, \emph{not all} the multi-trace expectations factor at large $N$ in terms of the single-trace ones and the Gaussian scaling is \emph{not} subadditive on the connected components. This is in stark contrast to the $D=2$ case of random matrices in which the multi-trace expectations always factor at large $N$. The best one can do for $D\ge 3$ is to identify restricted families of invariants for which the large $N$ factorization holds and we check that this indeed happens when restricting to the family of  melonic observables, the dominant family in the large $N$ limit.
\end{abstract}

\hrule\bigskip


\section{Introduction}

Initially introduced as generating functions for random discrete geometries \cite{ambj3dqg,sasa1}, random tensors \cite{RTM} have been increasingly studied as an interesting generalization of random matrices. Like random matrices, random tensors exhibit a large $N$ limit \cite{RTM}, but, contrary to random matrices, this limit is dominated by \emph{melonic} \cite{RTM,critical,Benedetti:2017qxl,Klebanov:2016xxf,Carrozza:2018ewt,Carrozza:2021qos,Prakash:2019zia}, rather than planar diagrams.
A second set of applications of such models in physics stems from the realization that
quantum mechanical models built from random tensors
are similar to the Sachdev-Ye-Kitaev model \cite{Sachdev:1992fk,Kitaev,Maldacena:2016hyu} but do not require quenching the disorder \cite{Witten:2016iux,Ferrari:2019ogc,Pakrouski:2018jcc,Klebanov:2018nfp}. Higher dimensional tensor field theories constitute a new class of large $N$ field theories \cite{Giombi:2017dtl,Giombi:2018qgp,Kim:2019upg,Choudhury:2017tax,Benedetti:2018ghn,Benedetti:2020iku,Benedetti:2019eyl} and led to new families of conformal field theories \cite{Gurau:2019qag,Klebanov:2018fzb,Benedetti:2020seh}.

From a mathematical point of view, tensors are quite different from matrices. For instance, while tensor eigenvalues exist and one can study their limit large $N$ distribution \cite{Gurau:2020ehg,Evnin:2020ddw,Majumder:2024ntn,Delporte:2024izt,Kloos:2024hvy,Bonnin:Wig1}, one needs to account for the fact that a tensor typically has exponentially many eigenvalues and eigenvectors \cite{cartwright2013number,breiding2019many} and that the eigenvector corresponding to an eigenvalue is unique (the eigenvalue equation is an inhomogenous polynomial equation in the eigenvector components). For random tensors, this leads to a novel eigenvector landscape \cite{arous2019landscape}.

More recently, there has been an increased effort to find an appropriate generalization of free probability theory~\cite{Voi1,Voi2} to random tensors \cite{Collins:2020iri, Collins:2022akx,Bonnin:2024lha,Collins:2024pip,Bonnin:free,Au:2023} and in particular generalizations of free cumulants, and asymptotic moments -- free cumulants relations for tensors have been proposed \cite{Bonnin:2024lha,Collins:2024pip}.

\paragraph{Large $N$ factorization for random matrices.} One main feature of random $N\times N$ matrices is that, for a large class of matrix probability measures,
in the large $N$ limit the moments (that is, the expectations of products of matrix invariants) factorize into products of expectations of the individual invariants. For instance, for a Hermitian random matrix $H$ one has
\[
 \Braket{\Tr(H^{n_1}) \dots \Tr(H^{n_q})}
  \sim_{N \to \infty}  \Braket{\Tr(H^{n_1}) } \ldots \Braket{\Tr(H^{n_q})} \; .
\]
In turn, the expectations of the individual invariants are equal to the connected expectations, also known as classical cumulants in the mathematics literature\footnote{In the free probability literature in particular it is common to denote this cumulant as $\varphi_{n_1}(H)$. We decide here to stick to the physics notation $\Braket{\ldots}_{\rm con.}$ for the connected moments.}
\[\Braket{\Tr(H^{n_1}) } = \Braket{\Tr(H^{n_1}) }_{\rm con.} \; .\]

When approaching the free probability theory starting from matrices of increasing size, this factorization allows one to completely forget about the moments and focus instead on the classical cumulants.  One promotes the classical cumulants to \emph{asymptotic moments} and subsequently identifies further a new class of cumulants, the \emph{free cumulants}. The two are
related in the strict large $N$ limit by novel moments cumulant relations in the lattice of non crossing partitions \cite{NicaSpeicher200609}.

The large $N$ factorization is also crucial for understanding large $N$ conformal field theories, which play a fundamental role in the AdS/CFT correspondence \cite{Maldacena:1997re}. The factorization ensures that in the strict large $N$ limit the theory becomes classical (free in the physics sense, that is the quantum fluctuations are suppressed) and the operator product expansion significantly simplifies \cite{Heemskerk:2009pn}.

\paragraph{The main contribution of this paper.}
Seen how crucial the large $N$ factorization is for random matrices, one can wonder if a similar factorization property holds for random tensors.
In fact, a conjecture was formulated in \cite{Collins:2024pip} stating that this should indeed be the case for the simplest case of a Gaussian random tensor.

The result of this paper is that this conjecture is wrong: even in the case of Gaussian random tensors, factorization \emph{does not hold} in general. It is not true that an arbitrary expectation of a product of invariants will factorize as a product of individual expectations.
So where does this leave us? The main lesson and message of this paper is that in random tensors one needs to be careful. When discussing the asymptotic behavior of the expectations of observables at large $N$, the factorization will only hold for subfamilies of observables. 
We identify a first such family in this paper, namely the melonic family, and we expect that many more useful families can be found. However, in the full theory, the large $N$ factorization fails and the structure of the large $N$ limit is more intricate.

\section{Notation and main theorem}

\paragraph{Tensors and trace invariants.} We use the notation of
\cite{RTM} and we discuss the case of real tensors (the complex case is similar). A real tensor $T_{a^1a^2\dots a^D}$ living in the external tensor product of $D$ fundamental representations of the orthogonal group $O(N)$ changes under a change of basis as
\begin{equation}\label{eq:basechange}
T'_{a^1a^2\dots a^D} = \sum_{b} O^1_{a^1b^1}O^2_{a^2b^2} \dots  O^3_{a^Db^D} \;
 T_{b^1b^2\dots b^D} \;, \qquad O^1,O^2,\ldots, O^D \in O(N) \; .
\end{equation}
We emphasize that each tensor index transforms with its own orthogonal transformation and that the tensor has no symmetry properties under permutations of its indices.

An \emph{edge $D$-colored graph} is a graph with $2n$ vertices labelled $1,2,\dots , 2n$ and an edge set that is partitioned into 
$D$ perfect matchings\footnote{A perfect matching is a partition of the vertex set into two element subsets.} $M^1,M^2\ldots, M^D$ between the vertices.
We define the multi-set $\mathbf{M} =\{M^1,M^2,\ldots, M^D \}$, and we denote the corresponding graph, somewhat abusively, also by $\mathbf{M}$. A class of invariant polynomials in the tensor entries,
called \emph{trace invariants}, is associated to such edge $D$-colored graphs\footnote{In the complex case the graphs are bipartite. Mutatis mutandis, our results carry over to that case.}. The polynomial $\Tr_{\mathbf{M} }(T)$ associated with the graph $\mathbf{M}$ is built as follows \cite{RTM}:
\begin{itemize}
 \item to every vertex $i=1,\ldots,2n$ of $\mathbf{M}$ one associates a tensor $T_{a^1_i a^2_i \ldots, a^D_i}$;
 \item for every color $c=1,2,\ldots , D$ and every edge $\{k,l\}\in M^c$, one associates a Kronecker delta symbol $\delta_{a^c_k a^c_l}$ with $\{k,l\}$ identifying the indices in the position $c$ of the tensors corresponding to the end vertices $k$ and $l$;
\item one sums over all the available indices
\[
\Tr_{\mathbf{M} }(T)
 = \sum_{a} \left( \prod_{i=1}^{2n} T_{a^1_i a^2_i \ldots a^D_i}  \right)
  \prod_{c=1}^D \left( \prod_{\{ k,l \} \in M^c}
 \delta_{a^c_k a^c_l} \right) \; ,
\]
and obtains a polynomial which is invariant under a change of basis as in~\eqref{eq:basechange}.
\end{itemize}

For $D=2$, the trace invariants are just traces of powers of the product
of $T$ and its transposed
\[\Tr[ (T T^t  )^{n_1} ] \ldots \Tr[ (T T^t  )^{n_q} ] \; .\]
For any $D$, the trace invariants (also known as tensor networks in other corners of the literature \cite{Orus:2018dya}) form an over complete system at finite $N$, and a basis
in the $N\to \infty$ limit \cite{BenGeloun:2013lim,BenGeloun:2017vwn}, in the space of invariants. We observe that if a graph $\mathbf{M}$ consists of two connected components\footnote{We denote $\mathbf{M} = \mathbf{M}_1 \sqcup \mathbf{M}_2$ as the graph $\mathbf{M}$ that consists of the disjoint union of two connected components $\mathbf{M}_1$ and $\mathbf{M}_2$. Note that some relabeling of the vertices might be needed when taking such disjoint unions.} $\mathbf{M} = \mathbf{M}_1 \sqcup \mathbf{M}_2$ then the associated invariant splits accordingly
\[
 \Tr_{ \mathbf{M}_1 \sqcup \mathbf{M}_2} (T) =  \Tr_{ \mathbf{M}_1 } (T)  \;  \Tr_{ \mathbf{M}_2} (T) \;.
\]

\paragraph{Gaussian random tensors.} We are interested in the Gaussian expectations of trace invariants, which we denote by
\[
 \Braket{\Tr_{\mathbf{M}}(T)}
  = \int (\prod_{a^1,\ldots, a^D} dT_{a^1a^2 \ldots a^D} )\; \;
  e^{- \frac{N^{\nu}}{2}\sum_{a} ( T_{a^1a^2 \ldots a^D} )^2 }
   \;
  \Tr_{\mathbf{M} }(T) \; .
\]
The tensor entries are independent normally distributed random variables with mean $0$ and standard deviation $N^{-\nu}$, hence with covariance matrix
\[
 \Braket{T_{a^1a^2 \ldots a^D} T_{b^1b^2 \ldots b^D}  } = \frac{1}{N^\nu}\delta_{a^1b^1} \delta_{a^2b^2} \ldots
\delta_{a^Db^D} \; .
\]
The factor $N^\nu$ in the Gaussian measure is purely conventional, and we will see below that a common choice is $\nu = D-1$. We stress that our results are independent of this factor.

The Gaussian expectation of a product of $2n$ tensor entries is computed using the Wick theorem~\cite{universality} as a sum over all the possible pairings (perfect matchings, or Gaussian contractions in the physics literature) between the tensor entries of a product of one covariance for each pair in the matching
\[
\Braket{ \prod_{i=1}^{2n}T_{a_i^1a_i^2 \ldots a_i^D} } =
\frac{1}{N^{\nu n}} \sum_{M^0\in {\cal M}_n}
\prod_{\{k,l\} \in M^0}
\delta_{a_k^1a_l^1} \delta_{a_k^2a_l^2}
\ldots \delta_{a_k^Da_l^D} \; ,
\]
where $M^0$ runs over the set
$ {\cal M}_n$ of all possible perfect matchings on $2n$ elements.
Using $\sqrt{2\pi n} \; n^ne^{-n} < n!<\sqrt{2\pi n} \; n^ne^{-n} e^{\frac{1}{12n}} $, it is straightforward to see that $|{\cal M}_n| = \frac{(2n)!}{2^nn!}< \sqrt{2e} \frac{2^n}{e^n} n^n$. For the Gaussian expectation of an invariant one obtains
\[
\Braket{\Tr_{\mathbf{M} } (T) } =
\frac{1}{N^{\nu n}}
  \sum_{M^0\in {\cal M}_n}
  \sum_{a}
  \left( \prod_{c=1}^D  \prod_{ \{ k,l \} \in M^c}
 \delta_{a^c_k a^c_l} \right)
 \prod_{\{k,l \} \in M^0}
\delta_{a_k^1a_l^1} \delta_{a_k^2a_l^2}
\ldots \delta_{a_k^Da_l^D}
  \; .
\]

We note that the matchings $M^1,\ldots, M^D$ and $M^0$ play distinct roles: while the $M^c$ matchings identify only one index pair between the end tensors, the matching $M^0$ identifies all the indices between the two end tensors.
By referring to $F_n(M^0,M^i)$ as the number of cycles with edges alternating between $M^0$ and $M^i$ (and the index $n$ signals that there are $2n$ vertices in total), we conclude
\begin{equation}\label{eq:gaussexpec}
\Braket{\Tr_{\mathbf{M} } (T) } =
\frac{1}{N^{\nu n}}
  \sum_{M^0\in {\cal M}_n}  \;
  N^{F_n(M^0,M^1) + F_n(M^0,M^2) +\ldots + F_n(M^0,M^D) } \; .   
\end{equation}
For convenience, we define
$F_n(M^0, \mathbf{M}) = F_n(M^0,M^1) + F_n(M^0,M^2) + \ldots + F_n(M^0,M^D)$.
Let $G(M^0, \bM)$ be the graph with the same vertex set as $\bM$ and edge set $\bM\cup M^0$.
Then $G(M^0, \bM)$ is an edge $(D+1)$-colored graph (on $2n$ vertices).

The Gaussian scaling of the expectation of an observable is the maximal possible scaling with~$N$ of a term in~\eqref{eq:gaussexpec}; that is,
$\max_{M^0\in {\cal M}_n} F_n(M^0, \mathbf{M})$.
As we need to take the maximum over all potential choices for $M^0$, the Gaussian scaling is not easy to compute in general.

\paragraph{Melonic graphs.} The family of graphs that maximize the Gaussian scaling is well understood \cite{RTM,critical,GurSch}.

\begin{definition} [Melonic graphs \cite{RTM,critical,GurSch}]
Melonic graphs are edge $D$-colored graphs whose connected components are connected melonic graphs.
Connected melonic graphs with $D$ colors are
the family of graphs obtained by recursive insertions of two vertices connected by $D-1$ edges arbitrarily on any edge respecting the coloring, starting from the unique edge $D$-colored graph with two vertices and $D$ edges. This is depicted in Fig.~\ref{fig:melo} below.
\end{definition}

\begin{figure}[ht]
\begin{center}
\begin{tikzpicture}[scale=0.3]

\draw[black, very thick] (0,0) -- (6,0);
\draw[red, very thick] (0,0) to[out=90,in=90] (6,0);
\draw[blue, very thick] (0,0) to[out=-90,in=-90] (6,0);

\filldraw[black] (0,0) circle (8pt);
\filldraw[black] (6,0) circle (8pt);

\draw[black, very thick] (10,0) -- (16,0);
\draw[black, very thick] (10,2) -- (16,2);
\draw[red, very thick] (10,0) -- (10,2);
\draw[red, very thick] (16,0) -- (16,2);

\draw[blue, very thick] (10,0) to[out=-90,in=-90] (16,0);
\draw[blue, very thick] (10,2) to[out=90,in=90] (16,2);

\filldraw[black] (10,0) circle (8pt);
\filldraw[black] (16,0) circle (8pt);
\filldraw[black] (10,2) circle (8pt);
\filldraw[black] (16,2) circle (8pt);

\draw[black, very thick] (20,0) -- (26,0);
\draw[black, very thick] (20,2) -- (22,2);
\draw[black, very thick] (24,2) -- (26,2);

\draw[red, very thick] (20,0) -- (20,2);
\draw[red, very thick] (26,0) -- (26,2);
\draw[red, very thick] (22,2) to[out=-90,in=-90] (24,2);

\draw[blue, very thick] (20,0) to[out=-90,in=-90] (26,0);
\draw[blue, very thick] (20,2) to[out=90,in=90] (26,2);
\draw[blue, very thick] (22,2) to[out=90,in=90] (24,2);

\filldraw[black] (20,0) circle (8pt);
\filldraw[black] (26,0) circle (8pt);
\filldraw[black] (20,2) circle (8pt);
\filldraw[black] (26,2) circle (8pt);
\filldraw[black] (22,2) circle (8pt);
\filldraw[black] (24,2) circle (8pt);

\end{tikzpicture}
\end{center}
 \caption{Some melonic graphs with three colors.
 Left: the unique graph with two vertices and three edges. Center: the melonic graph obtained by inserting two vertices connected by two edges on the red edge of the graph on the left. Right: a melonic graph obtained after two insertions.
 }\label{fig:melo}
\end{figure}
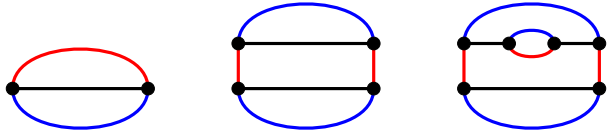

There is the following upper bound on the Gaussian scaling of an observable.

\begin{proposition}[Gaussian scaling bound \cite{RTM,universality}]\label{prop:histo} Let $D\ge 2$.
Let $\bM$ be an edge $D$-colored graph on $2n$ vertices.
If $\bM$ is connected, then for any $M^0$, we have
 \[
  F_n(M^0, \mathbf{M}) = 1 + (D-1)
   n -\omega(M^0, \mathbf{M}) \;,
   \qquad \omega(M^0, \mathbf{M}) \ge 0 \; ,
 \]
and
\begin{itemize}
 \item for $D=2$ and any $\mathbf{M} $, 
 there exist exactly $\frac{1}{n+1}\binom{2n}{n}$ matchings $M^0$ such that $ \omega(M^0, \mathbf{M}) =0$;
 \item for $D \ge 3$, we have $ \omega(M^0, \mathbf{M}) =0$ if and only if
  $G(M^0,\mathbf{M})$ is a melonic graph with $D+1$ colors and in this case $\mathbf{M}$ is a (connected) melonic graph with $D$ colors and there exists exactly one $M^0$ such that $ \omega(M^0, \mathbf{M}) =0$.
\end{itemize}
If  $\mathbf{M} $ has $q>1$ connected components and $G(M^0,\mathbf{M})$ is connected, then
\[
   F_n(M^0, \mathbf{M}) = D - (D-1) q + (D-1)
   n -\omega(M^0, \mathbf{M}) \;,
   \qquad \omega(M^0, \mathbf{M}) \ge 0 \;,
\]
and
\begin{itemize}
 \item for $D=2$ and any $\mathbf{M}$, there exist several $M^0$ such that $ \omega(M^0, \mathbf{M}) =0$;
 \item for $D\ge 3$, we have $ \omega(M^0, \mathbf{M}) =0$ if and only if $G(M^0,\mathbf{M})$ is a melonic graph with $D+1$ colors and in this case $ \mathbf{M}$ is a (disconnected) melonic graph and there exist several $M^0$ such that $G(M^0,\mathbf{M})$ is connected and melonic.
\end{itemize}
\end{proposition}

This proposition explains in particular why one usually takes the scaling $\nu = D-1$ for the Gaussian. 
With this scaling, for every connected $ \mathbf{M}$, we have
\[
 \Braket{\frac{1}{N} \Tr_{\mathbf{M} } (T) } = \sum_{M^0\in {\cal M}_n} N^{-\omega(M^0, \mathbf{M} )} \;;
\]
that is, the rescaled expectation of
any connected trace invariant is a series in $1/N$. Furthermore, the observables with maximal possible scaling (leading observables) are the melonic ones.

\paragraph{Classical cumulants.}
The \emph{classical cumulants}, also called connected moments of the Gaussian measure, are defined starting from the expectations. 
In the following, we refer to $\pi$ as a set partition of the set $\{1,\ldots, q\}$, that is, $\pi$ is a set of disjoint subsets of $\{1,\ldots, q\}$ whose union is $\{1,\ldots, q\}$. Furthermore, we denote by $\le$ the refinement order among partitions and by $1_q$ the one set partition. 
If $\mathbf{M}_1,\ldots, \mathbf{M}_q$ are connected graphs, 
the connected moments $\Braket{\ldots}_{\rm con.}$, or classical cumulants $\varphi_{\ldots}\big( \ldots \big)$ in the free probability literature, are defined implicitly by the relation
\begin{equation}\label{eq:momcum}
\Braket{ \Tr_{\mathbf{M}_1}(T) \ldots \Tr_{\mathbf{M}_q}(T) } =
\sum_{ \pi \le 1_q  }
\prod_{B\in \pi}  \Braket{  \prod_{i\in B}  \Tr_{\mathbf{M}_{i}}(T) }_{\rm con.}\; ,
\end{equation}
which is inverted using the M\"obius inversion to
\[
\Braket{ \Tr_{\mathbf{M}_1}(T) \ldots \Tr_{\mathbf{M}_q}(T) }_{\rm con.}=
\sum_{ \pi \le 1_q  } \lambda_{\pi}
\prod_{B\in \pi}  \Braket{ \prod_{i\in B}  \Tr_{\mathbf{M}_{i}}(T) }\; ,
\]
where $\lambda_\pi = (-1)^{|\pi|-1}(|\pi|-1)!$ is the M\"obius function in the lattice of partitions. 
For a connected component $\mathbf{M}_1$ only the one set partition contributes in~\eqref{eq:momcum}. 
Hence the connected expectation (classical cumulant) equals the expectation
\[
  \Braket{ \Tr_{\mathbf{M}_1}(T) } =   \Braket{ \Tr_{\mathbf{M}_1}(T) }_{\rm con.}  \; .
\]
For two connected graphs $\mathbf{M}_1$ and $\mathbf{M}_2$, the 
two partitions $\{ \{1\}, \{2\} \}$ and $\{\{1,2\}\}$ contribute
\[
\Braket{ \Tr_{\mathbf{M}_1}(T) \Tr_{\mathbf{M}_2}(T) } =
 \Braket{  \Tr_{\mathbf{M}_1}(T) }_{\rm con.} \,  
 \Braket{ \Tr_{\mathbf{M}_2}(T) }_{\rm con.} +  
 \Braket{ \Tr_{\mathbf{M}_1}(T) \,  \Tr_{\mathbf{M}_2}(T) }_{\rm con.}  \;.
\]

Classical cumulants are defined for any probability measure. For the Gaussian measure, it is straightforward to show that any classical cumulant is computed as a sum over matchings $M^0$ such that the graph $G(M^0, \mathbf{M} )$ is connected
\begin{equation}\label{eq:cumGaus}
 \Braket{ \Tr_{\mathbf{M} } (T) }_{\rm con.} =
\frac{1}{N^{\nu n}}
  \sum_{\substack{ {M^0\in {\cal M}_n } \\
  { G(M^0, \mathbf{M} ) \text{ connected} } } }  \;
  N^{F_n(M^0, \mathbf{M}) } \; .
\end{equation}
If $\mathbf{M} $ is connected to begin with, then the connectivity restriction on $G(M^0, \mathbf{M} )$ is automatically fulfilled  and one finds that, as already highlighted, the moment and classical cumulant for a connect $\mathbf{M}  $ coincide.
But this restriction does play a role if $\mathbf{M} $ has several connected components $\mathbf{M} = \mathbf{M}_1 \sqcup \ldots \sqcup \mathbf{M} _q$. The Gaussian scaling of a classical cumulant (connected moment) is the maximal scaling
\[
 \max_{\substack{ {M^0\in {\cal M}_n } \\
  { G(M^0,  \mathbf{M}_1 \sqcup \ldots \sqcup \mathbf{M} _q ) \text{ connected} } } }   F_n(M^0,  \mathbf{M}_1 \sqcup \ldots \sqcup \mathbf{M} _q )  \; ,
\]
of a term in the expansion~\eqref{eq:cumGaus}.

\paragraph{Large $N$ factorization.}
The large $N$ factorization of the expectations occurs when in~\eqref{eq:momcum} the only  leading contribution at large $N$ is given by the partition of the set $\{1,\ldots q\}$ into one element sets $\pi = \{ \{1\},\{2\}, \ldots \{q\} \}$ and all the other terms are strictly suppressed in $1/N$, that is,
\[
 \Braket{ \Tr_{\mathbf{M}_1}(T) \ldots \Tr_{\mathbf{M}_q}(T) } =
 \prod_{\rho=1}^q
 \Braket{  \Tr_{\mathbf{M}_{\rho}}(T) }_{\rm con.}  \; \bigg(1 + O(N^{-1}) \bigg) =
  \prod_{\rho=1}^q
 \Braket{ \Tr_{\mathbf{M}_{\rho}}(T) }  \; \bigg(1 + O(N^{-1}) \bigg) 
 \;  .
\]
In order for this to hold, the Gaussian scaling of the cumulants needs to be strictly subadditive; that is, 
for any connect $\mathbf{M}_1,\ldots \mathbf{M}_q$ with $2n_1,2n_2,\ldots 2n_q$ vertices, respectively, we must have
\begin{equation}\label{label:subadditive}
 \max_{\substack{ { M^0 \in {\cal M}_{n_1+\ldots +n_q} } \\ { G (M^0, \mathbf{M}_1 \sqcup \ldots \sqcup \mathbf{M}_q ) \text{ connected} } } }  
 F_{n_1+\ldots + n_q} (M^0, \bM_1 \sqcup \ldots \sqcup \bM_q ) 
  <
  \sum_{\rho=1}^q \max_{ M^0_\rho \in {\cal M}_{n_\rho}  }  F_{n_1}(M^0_\rho , \bM_\rho )
\end{equation}
where $M^0_\rho$ are perfect matchings on the $2n_\rho$ vertices of $\mathbf{M}_\rho$, and on the right hand side all the $G(M^0_\rho , \mathbf{M}_\rho ) $ are connected as each $\mathbf{M}_\rho$ is connected.

\paragraph{The large $N$ factorization holds for $D=2$ and for the melonic family in any $D$.}
The large $N$ factorization always holds when restricting to the melonic family. 
Indeed, for connected melonic observables $\mathbf{M}_\rho$ with $2n_\rho$ vertices the Gaussian scaling of the cumulants can be read off from Proposition~\ref{prop:histo} and
\[
 D - (D-1)q + (D-1) ( n_1 + \ldots +n_q ) <
  \sum_{\rho=1}^q ( 1 + (D-1) n_\rho ) \;, \;\; q\ge 2
  \; .
\]

Moreover, as the bounds in
Proposition~\ref{prop:histo} are saturated by all the observables in $D=2$, one recovers the known results that for $D=2$ all the multi-trace expectations factor at large $N$ into products of expectations of the single trace observables.

\paragraph{The general case.} For $D\ge 3$ one is naturally led to ask the question whether the large $N$ factorization holds for arbitrary choices of observables. The answer to this question is negative and proving this fact is our main result.

\begin{theorem} [Main Theorem]\label{thm:main}
 The Gaussian scaling for random tensors with $D\ge 3$ indices is \emph{not} strictly subadditive. In particular, there exist choices of connected graphs $\mathbf{M}_1, \mathbf{M}_2,  \ldots, \mathbf{M}_q$ such that
\begin{align*}
\max_{\substack{ { M^0 \in {\cal M}_{n_1+\ldots +n_q}} \\ { G(M^0, \mathbf{M}_1 \sqcup \ldots \sqcup \mathbf{M}_q ) \text{ connected} } } }  F_{n_1+\ldots + n_q} (M^0, \mathbf{M}_1 \sqcup \ldots \sqcup \mathbf{M}_q ) 
  >
  \sum_{\rho=1}^q\max_{ M^0_\rho \in {\cal M}_{n_\rho}  }  F_{n_\rho}(M^0_1 , \mathbf{M}_\rho ) \;,
\end{align*}
that is, the joint cumulant (connected expectation) of a multi-trace observable is larger in scaling than the product of the
single trace expectations.

More precisely, for every $\epsilon>0$ and for $n$ large enough such that 
$n> \frac{\ln(  \sqrt{2e} /\epsilon ) }{ \ln(e/2)}$ and
$\frac{(D-2)n}{2} >  (D \ln 7 - \ln 2) \, \frac{n}{\ln n} + D $ hold, a fraction of at least 
$1-\epsilon$ of all the edge $D$-colored graphs $\mathbf{M}$ on $2n$ vertices have at least a connected component $\mathbf{M}_{\rho}$ with $2n_\rho$ vertices such that
\[
  \max_{\substack{ { M^0 \in {\cal M}_{2n_\rho } } \\ { G(M^0, \mathbf{M}_\rho \sqcup \mathbf{M}_\rho) \text{ connected} } } }  F_{2n_\rho} (M^0, \mathbf{M}_\rho\sqcup \mathbf{M}_\rho ) \ge Dn_\rho > 
  2\max_{ M^0 \in {\cal M}_{n_\rho}  }  F_{n_\rho}(M^0 , \mathbf{M}_\rho ) 
  \;.
\]
Therefore, the connected expectation
$ \Braket{ \Tr_{ \mathbf{M}_\rho} (T) \Tr_{ \mathbf{M}_\rho} (T)  }_{\rm con.} 
$ is strictly larger in scaling in $N$ than the product of expectations
$\Braket{ \Tr_{ \mathbf{M}_\rho } (T) } \Braket{ \Tr_{ \mathbf{M}_\rho } (T) }$.
\end{theorem}

This theorem does not apply for $D=2$, 
because in that case the inequality $\frac{(D-2)n}{2} >  (D \ln 7 - \ln 2) \, \frac{n}{\ln n} + D $ false for all large enough $n$.
For $D\ge 3$, the inequality is always satisfied for large enough $n$; 
for $D=3$ for example, starting at an $n$ of order $\sim 7^6$. 

\section{Proof of the main theorem}

In this section, we prove our main theorem.
Our strategy is to show that if we select $D$ perfect matchings at random,
then with probability very close to $1$, 
they form a collection that verifies the statement of Theorem~\ref{thm:main}.

\paragraph{The boundary graph.}
In the following, we will introduce \emph{boundary graphs} \cite{RTM,universality}.
These graphs arise from edge $D$-colored graphs $\bM$ when we consider a not necessarily perfect matching $\bar M^0$
and encode in this graph all information that is already given by $\bar M^0$ in view of $\bar M^0$ being completed to a perfect matching $M^0$ at some later stage.

To this end, let $\mathbf{M}=\{M^1,\ldots,M^D\}$ be an edge $D$-colored graph on $2n$ vertices, and consider a matching $\bar M^0$ with at most $n$ edges on the vertex set of $\bM$.
Observe that $M^c\cup \bar M^0$ is a graph whose vertices are incident to at most two edges;
that is, its compenents consists of paths and even cycles.
In addition, the paths and cycles have edges alternating between $M^c$ and $\bar M^0$.
Crucially, for these paths the first and the last edge are always in $M^c$ (the first and the last edge may be the same).
Let us call these paths the alternating paths of $\bM$ with respect to $\bar M^0$
and we say a path has color $c$ if the path has a nonempty intersection with $M^c$.
The boundary graph $\partial_{\bar M^0}\bM$ of $\bM$ with respect to $\bar M^0$ arises from $\bM$
by adding for all colors $c$ and for all alternating paths with color $c$ an edge of color $c$ between the two end points of the paths; afterwards we delete all the original edges and vertices not incident to the new edges, that is, exactly all vertices of $\bar M^0$.
Note that the boundary graph is again an edge $D$-colored graph on $2n-2|\bar M^0|$ vertices.

Some examples are displayed in Figure~\ref{fig:boundary}. We note that it is possible for the graph $\mathbf{M}$ to be connected and have a disconnected boundary graph $\partial_{\bar M^0}\mathbf{M}$ for some $\bar M^0$.

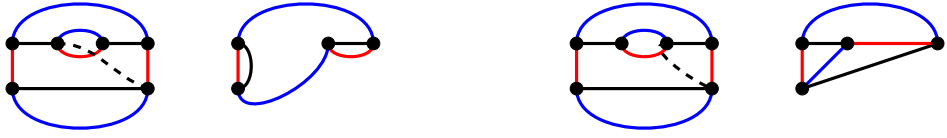
\begin{figure}[ht]
\begin{center}
\begin{tikzpicture}[scale=0.3]

\draw[black, very thick] (0,0) -- (6,0);
\draw[black, very thick] (0,2) -- (2,2);
\draw[black, very thick] (4,2) -- (6,2);

\draw[red, very thick] (0,0) -- (0,2);
\draw[red, very thick] (6,0) -- (6,2);
\draw[red, very thick] (2,2) to[out=-90,in=-90] (4,2);

\draw[blue, very thick] (0,0) to[out=-90,in=-90] (6,0);
\draw[blue, very thick] (0,2) to[out=90,in=90] (6,2);
\draw[blue, very thick] (2,2) to[out=90,in=90] (4,2);

\draw[black,dashed,very thick] (2,2) to[out=0,in=160] (6,0);

\filldraw[black] (0,0) circle (8pt);
\filldraw[black] (6,0) circle (8pt);
\filldraw[black] (0,2) circle (8pt);
\filldraw[black] (6,2) circle (8pt);
\filldraw[black] (2,2) circle (8pt);
\filldraw[black] (4,2) circle (8pt);

\draw[black, very thick] (10,0) to[out=0,in=0] (10,2);
\draw[black, very thick] (14,2) -- (16,2);

\draw[red, very thick] (10,0) -- (10,2);
\draw[red, very thick] (14,2) to[out=-90,in=-90] (16,2);

\draw[blue, very thick] (10,0) to[out=-90,in=-90] (14,2);
\draw[blue, very thick] (10,2) to[out=90,in=90] (16,2);

\filldraw[black] (10,0) circle (8pt);
\filldraw[black] (10,2) circle (8pt);
\filldraw[black] (16,2) circle (8pt);
\filldraw[black] (14,2) circle (8pt);

\begin{scope}[shift=({-5,0})]

\draw[black, very thick] (30,0) -- (36,0);
\draw[black, very thick] (30,2) -- (32,2);
\draw[black, very thick] (34,2) -- (36,2);

\draw[red, very thick] (30,0) -- (30,2);
\draw[red, very thick] (36,0) -- (36,2);
\draw[red, very thick] (32,2) to[out=-90,in=-90] (34,2);

\draw[blue, very thick] (30,0) to[out=-90,in=-90] (36,0);
\draw[blue, very thick] (30,2) to[out=90,in=90] (36,2);
\draw[blue, very thick] (32,2) to[out=90,in=90] (34,2);

\draw[black,dashed,very thick] (34,2) to[out=-180,in=160] (36,0);

\filldraw[black] (30,0) circle (8pt);
\filldraw[black] (36,0) circle (8pt);
\filldraw[black] (30,2) circle (8pt);
\filldraw[black] (36,2) circle (8pt);
\filldraw[black] (32,2) circle (8pt);
\filldraw[black] (34,2) circle (8pt);

\draw[black, very thick] (40,0) -- (46,2);
\draw[black, very thick] (40,2) -- (42,2);

\draw[red, very thick] (40,0) -- (40,2);
\draw[red, very thick] (42,2) -- (46,2);

\draw[blue, very thick] (40,0) -- (42,2);
\draw[blue, very thick] (40,2) to[out=90,in=90] (46,2);

\filldraw[black] (40,0) circle (8pt);
\filldraw[black] (40,2) circle (8pt);
\filldraw[black] (46,2) circle (8pt);
\filldraw[black] (42,2) circle (8pt);

\end{scope}

\end{tikzpicture}
\end{center}
 \caption{The first and the third graph display the same edge $3$-colored graph with two different choices for $\bar M^0$, where the dash edges indicate the edges of $\bar M^0$. The second and the fourth graph show the boundary graph with respect to the particular choice of $\bar M^0$.}\label{fig:boundary}
\end{figure}

\paragraph{A necessary and sufficient condition for factorization.}
We proceed with some observations regarding a potential factorization.
Imagine all multi-trace observables factorize then, for any connected
$\mathbf{M}$ on $2n$ vertices, the expectation
\[
\Braket{\Tr_{\mathbf{M}}(T) \Tr_{\mathbf{M}}(T)  } =
 \Braket{  \Tr_{\mathbf{M}}(T) }_{\rm con.} \Braket{  \Tr_{\mathbf{M}}(T) }_{\rm con.} + \Braket{  \Tr_{\mathbf{M}}(T) \Tr_{\mathbf{M}}(T) }_{\rm con.}\;,
\]
should also factorize; 
that is, the disconnected (first) term on the right-hand side in the above equation must be strictly larger in scaling than the connected (second) term.
Among the Gaussian pairings $M^0$ admissible for the connected term, 
one consists in connecting every vertex in one $\mathbf{M}$ to its copy in the second $\mathbf{M}$. 
This pairing always satisfies
$F_{2n}(M^0,\mathbf{M} \sqcup \mathbf{M} ) = D n$. 
It follows that the multi-trace observables factorize, 
or equivalently the Gaussian scaling of cumulants is strictly subadditive, 
only if for every connected $\mathbf{M}$, we have
\begin{equation}\label{label:Dn/2}
\max_{ M^0 \in {\cal M}_{n}  }  F_{n}(M^0 , \mathbf{M} ) > \frac{Dn}{2} \;.    
\end{equation}
In fact, the converse is also true as the following lemma shows.

We remark that~\eqref{label:Dn/2} holds for all connected edge $D$-colored graphs $\bM$ if and only if~\eqref{label:Dn/2} holds for all edge $D$-colored graphs $\bM$.
Indeed, if it holds for all edge $D$-colored graphs, then clearly it holds for all connected edge $D$-colored graphs.
If it holds for all connected graphs, then simply take a matching in each component that attains the maximum and then the union of these matchings shows that~\eqref{label:Dn/2} holds for all graphs.

\begin{lemma}\label{lem:reform}
The Gaussian scaling is strictly subadditive if and only if for any connected edge $D$-colored graph $\bM$ on $2n$ vertices, we have
 \[
\max_{ M^0 \in {\cal M}_{n}  }  F_{n}(M^0 , \mathbf{M} ) > \frac{Dn}{2} \;.
\]
\end{lemma}

\begin{proof}
As we argued above, if the Gaussian scaling is strictly subadditive, then~\eqref{label:Dn/2} holds.
Hence, it remains that the reversed implication is also true.

To this end, consider connected edge $D$-colored graphs $\mathbf{M}_1,\ldots ,\mathbf{M}_q$;
say, $\bM_\rho$ has $2n_\rho$ vertices and $\sum_\rho n_\rho=n$ and $\bM=\mathbf{M}_1 \sqcup \dots \sqcup \mathbf{M}_q$
as well as $\bM=\{M^1,\ldots,M^D\}$.
We need to check that~\eqref{label:subadditive} holds.
Let $M^0$ be any perfect matching on the vertex set of $\bM$ such that $G(M^0,\bM)$ is connected.
In order to finish the proof it suffices to construct perfect matchings $M_\rho$ on the vertex set of $\bM_\rho$ 
such that $F_n(M^0,\bM)< \sum_\rho F_{n_\rho}(M_\rho,\bM_\rho)$ holds.

Let us turn to the details.
For each $\rho$, we denote by $M_\rho^0$ the set of edges of $M^0$ with both endpoints in $\bM_\rho$.
Let $S^0=M^0 \setminus (M^0_1 \cup \ldots \cup M^0_q)$; 
that is, $S^0$ is the subset of $M^0$ consisting of edges that join two different $\bM_\rho$.
Note that $S^0\neq \emptyset$.

Next, we exploit that the boundary graphs $\partial_{M^0_\rho}\bM_{\rho}$ encode all important information we need for the edges in $M^0$ that join two different $\bM^\rho$.
Observe that $S^0$ is a perfect matching on $ \bigsqcup_{\rho} \partial_{M^0_\rho}\bM_{\rho}$
and let $s= |S|$.
Crucially, no edge of $S^0$ belongs to a cycle of with only two edges.
Thus the number of cycles that alternate between edges in $S^0$ and edges in some $M^c$
is at most $\frac{Ds}{2}$.

By our assumption, for each $\rho$, there exists a perfect matching $\overline M_\rho^0$ in $\partial_{M^0_\rho}\bM_{\rho}$ such that $\overline M^0=\bigcup_\rho \overline M_\rho^0$ yields
more than $\frac{Ds}{2}$ cycles that alternate between edges in $\overline M^0$ and edges in some $M^c$.
Therefore, $F_n(M^0,\bM)< \sum_\rho F_{n_\rho}(M^0_\rho \cup \overline M_\rho^0,\bM_\rho)$,
which completes the proof.
\end{proof}

In order to prove Theorem~\ref{thm:main}, it is therefore enough to exhibit a graph that violates the bound in Lemma~\ref{lem:reform}. It turns out that finding an explicit counterexample is not immediate, but one can use a statistical argument and prove that among large enough graphs, essentially all instances are counterexamples.

\paragraph{Additional families of graphs for which the bound in Lemma~\ref{lem:reform} holds.} 
Of course, the bound in Lemma~\ref{lem:reform} holds for melonic graphs. It holds, however, also for larger classes of graphs. 

Consider the case where $D=3$ and the graph $\bM=\{M^1,M^2,M^3\}$ is planar when embedded respecting the colors, 
that is, such that the faces, or 2-cells, of the embedded graph are the bi-colored cycles. We emphasize that this is the notion of planarity of ribbon graphs \cite{ellis2013graphs} (and combinatorial maps), and it is slightly different from the usual notion of planarity in graph theory. 
For ribbon graphs one requires that the two dimensional cellular complex with 0-cells (the vertices), 1-cells (the edges) and 2-cells (the prescribed faces) is homeomorphic to the 2-dimensional sphere. 
In the usual sense in graph theory a graph is called planar if the 1-complex with 0-cells (the vertices) and 1-cells (the edges) can be completed to some 2-complex homeomorphic to the 2-dimensional sphere by adding in some way 2-cells formed by cycles of edges.

Suppose $\bM$ has $2n$ vertices and hence $3n$ edges. 
We denote by $F_n(M^i,M^j)$ the number of bi-colored cycles formed by alternating edges in $M^i$ and $M^j$.
Then the Euler relation\footnote{The Euler relation for a general embedded graph is $\chi=V- E + F = 2q-k$, where $\chi$ is the Euler characteristic, $V$ the number of vertices (0-cells), $E$ the number of edges (1-cells), $F$ the number of faces (2-cells), $q$ the number of connected components and $k$ the non-orientable genus of the embedded graph.} for such a planar graph implies
\[
 F_n(M^1,M^2) + F_n(M^1,M^3) + F_n(M^2,M^3) = n + 2q \;,
\]
where $q$ is the number of connected components. In this case there exist two pairs of colors, say $12$ and $13$ so that $F_n(M^1,M^2) + F_n(M^1,M^3) > \frac{n}{2}+q $. 
Taking $M^0 = M^1$ yields
\[
 F_n(M^0,\mathbf{M}) =  F_n(M^1,M^1) + F_n(M^1,M^2) + F_n(M^1,M^3)  > n + \frac{n}{2} + q \;.
\]

Note, however, that one can not conclude from here that, when restricted to planar observables, the large $N$ factorization holds. Contrary to the melonic family, for which we have checked directly that the factorization holds, for the planar family we have only checked that the Gaussian scaling respects the bound in Lemma~\ref{lem:reform}. 
As we shall see in the following, there exist (non-planar) graphs that violate the bound in Lemma~\ref{lem:reform}. 
Hence, the large $N$ factorization does not hold in general, and there is no guarantee a priori that factorization holds for planar graphs: in the proof of Lemma~\ref{lem:reform} we need to assume the bound holds for \emph{all} the graphs in order to obtain the factorization.

\paragraph{Pairs of random matchings.} 
Later we consider random graphs that are the union of $D$ randomly chosen matchings.
To this end, we first investigate the cycle structure of two random matchings (see Proposition~\ref{prop:matching}).
For the proof we need the following inequality.
Suppose that we are given
$0\le p_1 \le p_2\dots \le p_n$ with $\sum_{i=1}^n p_i=1$ and $A_1 \ge A_2\dots  \ge A_n $. 
Then for any $i$ and $j$, we have $(p_j-p_i) (A_i-A_j) \ge 0$.
Consequently, we can compute 
\begin{equation*}
 0\le  \frac{1}{2}\sum_{i=1}^n \sum_{j=1}^n \frac{p_j-p_i}{n}(A_i-A_j)
  = \sum_{i=1}^n \sum_{j=1}^n
 \frac{p_j A_i - p_iA_i}{n}
 =\frac{1}{n} \sum_{i=1}^n A_i -
 \sum_{i=1}^n p_iA_i.
\end{equation*}
Rearranging yields
\begin{align}\label{eq:bound}
     \sum_{i=1}^n p_iA_i  \le \frac{1}{n}
 \sum_{i=1}^n A_i \;.
\end{align}

In the following, we refer to $\Pro$ as the probability and to $\Exp$ as the expectation.

Observe that the union of two perfect matchings is a union of cycles of even length.
In addition, there is a bijection between the union of two perfect matchings on $2n$ vertices and permutations on $2n$ elements with only even cycles. 
The following could also be deduced from existing results in the literature~\cite[see the proof of Proposition 8.2]{Lugo:2009}, but for completeness we give the short proof here.

\begin{proposition}\label{prop:matching}
Fix some perfect matching $M^0$ on $2n$ vertices and add a perfect matching $M$ uniformly at random on the same vertex set.
Then for any $t>0$, we have
\[
 \Pro\Big[ F_n(M^0, M) \ge t \Big] \le \frac{7^n}{(2n)^t} \;.
\]
\end{proposition}
\begin{proof}
The proof is split into three parts.
We first make a general remark on random perfect matchings
afterwards we calculate the probability that a particular edge of $M^0$ is contained in a cycle in $G(M^0,M)$ of length $2k$ for any $k$.
As a third step, we prove an upper bound on $\Exp[m^{ F_n(M^0, M) }]$ for any integer $m\geq 2n$.
This is turn yields the desired bound by utilizing Markov's inequality.

Let us turn to the details.
Recall that $M$ is a perfect matching chosen uniformly at random among all the perfect matchings on $\{1,\ldots,2n\}$.
We make two simply observations.
\begin{enumerate}[label=(\roman*)]
    \item\label{symm} By symmetry, the vertex $i$ is joined by $M$ to some vertex $j$ and $j$ is a uniformly chosen element of $\{1,\ldots,2n\}\setminus \{i\}$.
    \item\label{cond} Let $v_1,\ldots,v_{2\ell}\in \{1,\ldots,2n\}$ be distinct.
    The distribution of $M$ conditioned on $\{v_{2i-1},v_{2i}\}\in M$ for all $i=1,\ldots,\ell$ equals that of a uniformly chosen perfect matching on $\{1,\ldots,2n\}\setminus \{v_1,\ldots,v_{2\ell}\}$.
\end{enumerate}

Recall that $M\cup M^0$ is the union of cycles of even length.
Let $\{a_1,a_2\}\in M^0$ be arbitrary but fixed.
For any $k$, we denote by $p_k$ the probability that $\{a_1,a_2\}$ belongs to a cycle of length $2k$.
Thus, $p_k=0$ for all $k>n$.

By \ref{symm}, we conclude that the probability that $a_2$ is matched to $a_1$ in $M$ equals $1/(2n-1)=p_1$.
Suppose now $k\geq 2$.
In the following we tacitly assume that $a_i\neq a_j$ for $i\neq j$ if not stated otherwise.
The probability that $a_2$ is not matched to $a_1$ in $M$ equals $1-1/(2n-1)=(2n-2)/(2n-1)$, again by \ref{symm}.
So suppose $\{a_2,a_3\}\in M$ for some $a_3$.
In turn $\{a_3,a_4\}\in M^0$ for some $a_4$.
The probability that $a_4$ is matched to $a_1$ in $M$ (conditioned on $\{a_2,a_3\}\in M$) is $1/(2n-3)$, by \ref{symm} and \ref{cond}.
Hence $p_2=\frac{2n-2}{2n-1}\frac{1}{2n-3}$.

This easily reveals the general pattern.
Suppose $\{a_1,a_2\},\ldots,\{a_{2i-1},a_{2i}\}\in M^0$ and we condition on $\{a_2,a_3\},\ldots,\{a_{2i-2},a_{2i-1}\}\in M$,
then the probability that $a_{2i}$ is matched to $a_1$ in $M$ is $1/(2n-2i+1)$ and
the probability that $a_{2i}$ is matched to $a_{2i+1}$ in $M$ is $1-1/(2n-2i+1)=(2n-2i)/(2n-2i+1)$, again by \ref{symm} and \ref{cond}.
Consequently, this yields
\[
 p_k = \frac{2n-2}{2n-1} \; \frac{2n-4}{2n-3} \; \dots
 \frac{2n -2k+2}{2n -2k+3} \; \frac{1}{2n-2k+1} \; .
\]
Observe that 
\[
 p_{k+1} = p_k \frac{2n-2k}{2n-2k-1}  = p_k \left( 1 + \frac{1}{2n-2k-1}\right) \; ,
\]
for all $k<n$ and so $p_1 < \ldots < p_n$ as well as $\sum_{i=1}^np_i=1$.

Next, we prove by induction on $n$ that for any integer $m\geq 2n$
\[
\Exp[m^{ F_n(M^0, M) }] \le \binom{m+n-1}{m-1}
\]
holds. 
Indeed, fix an edge $e$ in $M^0$ and for $n=1$ we evaluate
$ \mathbb{E}[m^{ F_1(M^0, M) }] = m =\binom{m}{m-1}$. 
Let $\cE_k$ denote the event that $e$ belongs to a cycle of length $k$ in $M\cup M^0$.
For $n\ge 2$, we compute
\[
 \mathbb{E}[m^{ F_n(M^0, M) }] 
 = \sum_{k=1}^n p_k \; \mathbb{E}[ \, m^{ F_n(M^0, M) } | \; \cE_k \, ]
  = \sum_{k=1}^n p_k m \; \mathbb{E}[m^{ F_{n-k}(M^0_{n-k}, M_{n-k}) } ] \;,
\]
where $M^0_{n-k}$ is a matching on $2n-2k$ vertices and $M_{n-k}$ is a uniformly chosen perfect matching on the same vertex set;
for the last equality we used again \ref{cond}.
In addition, we set $ \Exp[m^{ F_{0}(M^0_{0}, M_{0}) } ] = 1$.

Recall that $\binom{r}{s}=\binom{r+1}{s+1}-\binom{r}{s+1}$ folds for all integers $r,s\geq 0$
and hence 
\[
\sum_{i=0}^s \binom{r+i-1}{r-1}
= \sum_{i=0}^s \left[ \binom{r+i}{r}- \binom{r+i-1}{r} \right]
= \binom{r+s}{r}-\binom{r-1}{r}
=  \binom{r+s}{r}.
\]
As $m\geq 2n$ implies that $p_1m\geq 1$, we observe that $\Exp[m^{ F_{n}(M^0_{n}, M_{n}) } ]\geq p_1m \Exp[m^{ F_{n-1}(M^0_{n-1}, M_{n-1}) } ]\geq \Exp[m^{ F_{n-1}(M^0_{n-1}, M_{n-1}) } ]$.
Hence $\Exp[m^{ F_{n}(M^0_{n-k}, M_{n-k}) } ]$ is decreasing in $k$.
This allows us to utilize~\eqref{eq:bound} and the induction hypothesis to conclude that
\begin{align*}
\mathbb{E}[m^{ F_n(M^0, M) }]  
&\le\frac{m}{n} \sum_{k=1}^n  \mathbb{E}[m^{ F_{n-k}(M^0_{n-k}, M_{n-k}) } ] \\
&\le    \frac{m}{n} \sum_{k=0}^{n-1}
\binom{m+k-1}{m-1}
=\frac{m}{n} \binom{m+n-1}{m}
=\binom{m+n-1}{m-1}  \;.
\end{align*}

For any real random variable $X\geq 0$ and any real number $t>0$,
we observe that $\Exp [X]/t\geq \Exp[1_{X\geq t}]=\Pro[X\geq t]$.
This is known as Markov's inequality.
Therefore, we obtain
\[
  \Pro \Big[ F_n(M^0,M) \ge t \Big] =
    \Pro \Big[ (2n)^{ F_n( M^0, M) } \ge (2n)^t \Big] \le
    \frac{ \Exp [ (2 n)^{ F_n( M^0, M)} ] }{(2n)^t} 
    \le \frac{\binom{3n-1}{n}}{(2n)^t}
    \leq \frac{7^n}{(2n)^t} \;,
\]
which completes the proof.\footnote{Using a straightforward induction argument yields $\binom{3n-1}{n} \le \frac{3^{3n}}{2^{2n}}$, marginally better than the $7^n$ bound we used.}
\end{proof}

\paragraph{Random graphs.} We are finally in the position to prove that among large enough graphs we always find some which violate the inequality in Lemma \ref{lem:reform}.

\begin{proposition}\label{prop:bound}
 There exists an edge $D$-colored graph $\mathbf{M}$ on $2n$ vertices such that 
 \[
 F_n(M^0,\mathbf{M}) < n +  (D\ln 7 - \ln 2)\, \frac{n}{\ln n} + D=(1+o(1))n \; ,
 \]
for every perfect matching $M^0$.
In fact, for every $\epsilon>0$ and $n> \frac{\ln( \sqrt{ 2e} /\epsilon ) }{ \ln(e/2)}$, a fraction of at least $1-\epsilon$ of all the edge $D$-colored graphs with $2n$ vertices satisfies this bound for every~$M^0$.
 \end{proposition}
\begin{proof}
We fix a perfect matching $M^0$ on the $2n$ vertices. 
We add uniformly at random the $D$ perfect matchings $M^1,\ldots, M^D$. 
We claim that
\[
 \mathbb{P} \Big[ F_n(M^0,\mathbf{M})
 \ge n+ (D\ln 7 - \ln 2) \, \frac{n}{\ln n} + D \Big] < n^{-n} \;.
\]
Indeed, this can be seen as follows. 
Let us denote $A = n + (D\ln 7 - \ln 2) \, \frac{n}{\ln n} + D$. 
We compute using Proposition~\ref{prop:matching} in the last line
\[
\begin{split}
&  \mathbb{P} \Big[ F_n(M^0,\mathbf{M}) \ge A \Big] \crcr
& \qquad \leq  \sum_{\substack{ {t_1,t_2, \ldots ,t_D\leq n } \\ {t_1+t_2+\dots +t_D = A}} }
 \mathbb{P} \Big[
F_n(M^0, M^1) \ge t_1 , F_n(M^0, M^2) \ge t_2, \dots
,F_n(M^0, M^D) \ge t_D \Big] \crcr
& \qquad \leq n^D \max_{\substack{ {t_1,t_2, \dots ,t_D} \\ {t_1+t_2+\dots +t_D = A}} }
\mathbb{P} \Big[
F_n(M^0, M^1) \ge t_1 , F_n(M^0, M^2) \ge t_2, \dots
,F_n(M^0, M^D) \ge t_D \Big]\crcr
& \qquad \le n^D \max_{\substack{ {t_1,t_2, \dots ,t_D} \\ {t_1+t_2+\dots +t_D = A}} } \;  \prod_{c=1}^D
\mathbb{P} \Big[ F_n(M^0, M^c) \ge t_c \Big]
\le \; n^D \prod_ {t_1+t_2+\dots +t_D = A} \frac{7^{n}}{(2n)^{t_c}}\crcr
&\qquad <
\; n^D \; 7^{Dn} \; 2^{-n}\; n^{-A} 
=n^{-n} \;.
\end{split}
\]
Observe that our random selection of $\bM$ corresponds to the uniform distribution on all edge $D$-colored graphs (on $2n$ labelled vertices).
It follows that a fraction larger than $1-n^{-n}$ of all the graphs $\bM$ satisfies  $F_n(M^0,\mathbf{M}) < A$. 
Using the union bound and taking into account that there are
in total $|{\cal M}_n|=\frac{(2n)!}{2^nn!}< \sqrt{2e} \frac{2^n}{e^n} n^n$ possible matchings  $M^0$, 
we conclude that there is a fraction of the graphs $\bM$ larger than $1 -|{\cal M}_n| n^{-n} > 1 - \sqrt{2e} \frac{2^n}{e^n} $ such that $F_n(M^0,\mathbf{M}) < A$ holds for all the matchings $M^0$.
\end{proof}

\paragraph{A multi-trace expectation that does not factorize.}
An explicit example of a multi-trace observable that does not factorize at large $N$ is obtained by
taking $n$ large enough such that $ n +  (D\ln 7 - \ln 2) \, \frac{n}{\ln n} + D  <\frac{D}{2}n $. 
If the graph  $\mathbf{M}$ in Proposition~\ref{prop:bound} is connected, 
then
$\Braket{\Tr_{\mathbf{M} }(T) \Tr_{\mathbf{M}}(T) }_{\rm con.} $ is strictly larger in scaling in $N$ than $\Braket{ \Tr_{\mathbf{M} }(T) } \Braket{ \Tr_{\mathbf{M}}(T) } $.
If $\mathbf{M}$ is disconnected, 
then for at least one of its connected components $\mathbf{M}_\rho$,
we have $\max_{M^0_{\rho}\in{\cal M}_{n_\rho}} F_{n_\rho}(M^0_\rho, \mathbf{M}_\rho) < \frac{Dn_\rho}{2}$ and
$\Braket{ \Tr_{\mathbf{M}_\rho }(T) \Tr_{\mathbf{M}_\rho}(T) }_{\rm con.} $  is strictly larger in scaling in $N$ than $\Braket{  \Tr_{\mathbf{M}_\rho }(T) } \Braket{  \Tr_{\mathbf{M}_\rho}(T) } $. 
Indeed, if for all the connected components of $ \mathbf{M}$ we had $\max_{M^0_{\rho}\in{\cal M}_{n_\rho}} F_{n_\rho}(M^0_\rho, \mathbf{M}_\rho) \ge \frac{Dn_\rho}{2}$,
then there would exist some matching $M^0$ such that 
$  F_n(M^0,\mathbf{M}) \ge \frac{Dn}{2}$, which contradicts the fact that for any 
$M^0$ we have $F_n(M^0,\mathbf{M}) <  n +  (D\ln 7 - \ln 2) \, \frac{n}{\ln n} + D < \frac{Dn}{2}$.

Combining this remark with Proposition~\ref{prop:bound} proves our main result, Theorem~\ref{thm:main}.

 \section*{Acknowledgments}

 This work is supported by the Deutsche Forschungsgemeinschaft (DFG, German Research Foundation) under
 Germany's Excellence Strategy EXC 2181/1 - 390900948
 (the Heidelberg STRUCTURES Excellence Cluster). Research of the third author was supported in part by SNSF grant 200021-228014.

\bibliographystyle{ieeetr}

\bibliography{Refs.bib}

\end{document}